\documentclass[letterpaper, 10pt, conference]{ieeeconf}      
\IEEEoverridecommandlockouts                              
\usepackage{amsmath}
\usepackage{amssymb}
\usepackage{dsfont}
\usepackage{graphicx}
\usepackage{mathtools}
\usepackage{theorem}
\usepackage{cite}
\usepackage{pstricks}
\usepackage{epsfig}
\usepackage{pst-grad} 
\usepackage{pst-plot} 
\usepackage{bbm}
\usepackage{float}
\usepackage{footnote}

{\theorembodyfont{\itshape}\newtheorem{theorem}{Theorem}}
{\theorembodyfont{\itshape}}
{\theorembodyfont{\itshape}}
{\theorembodyfont{\itshape}\newtheorem{corollary}{Corollary}}
{\theorembodyfont{\upshape}}
{\theorembodyfont{\itshape}}
{\theorembodyfont{\upshape}\newtheorem{remark}{Remark}}
{\theorembodyfont{\upshape}}
{\theorembodyfont{\upshape}}

\newcommand\numberthis{\addtocounter{equation}{1}\tag{\theequation}}

\newcommand{\1}{\mathds{1}}

\begin{document}
\title{\LARGE \bf 
A Communication-Free Master-Slave Microgrid with Power Sharing}


\author{Pooya~Monshizadeh \and Claudio De Persis \and Nima Monshizadeh \and Arjan van der Schaft
\thanks{Pooya Monshizadeh and Arjan van der Schaft are with the Johann Bernoulli Institute for Mathematics and Computer Science, University of Groningen, 9700 AK, the Netherlands,
        {\tt\small p.monshizadeh@rug.nl, a.j.van.der.schaft@rug.nl}}%
\thanks{Claudio De Persis and Nima Monshizadeh are with the Engineering and Technology Institute, University of Groningen, 9747 AG, the Netherlands,
        {\tt\small n.monshizadeh@rug.nl, c.de.persis@rug.nl}}%
\thanks{This work is supported by the STW Perspectief program "Robust Design of Cyber-physical Systems" under the auspices of the project "Energy Autonomous Smart Microgrids".}
}
\maketitle

\begin{abstract}
In this paper a design of a master-slave microgrid consisting of grid-supporting current source inverters and a synchronous generator is proposed. The inverters are following the frequency of the grid imposed by the synchronous generator. Hence, the proposed structure of the microgrid is steadily synchronized. We show that the method achieves power sharing without the need of communication. Furthermore, no change in operation mode is needed during transitions of the microgrid between islanded and grid-connected modes.
\end{abstract}


\section{Introduction}
A microgrid is a network of connected power sources and loads in a small area which can be seen as one entity within the wide area electrical power system. Considering such a network as a building block of the power grid is mainly motivated by preventing blackouts. The microgrid is capable of disconnecting itself from the main grid in case of a fault in the main grid and reconnecting when the fault is resolved. In classical networks, electrical energy sources were mainly synchronous generators (SG). Currently, however, to take advantage of storage systems and renewable energies such as wind, solar, geothermal, etc., an interface for power regulation or conversion from Direct Current (DC) to Alternating Current (AC) is also needed. These interfaces are called inverters. A microgrid typically includes both synchronous generators and inverters.

Inverters are divided into three categories. The first category, "grid-forming" inverters, act as a voltage source with fixed frequency and amplitude. In case the microgrid disconnects from the main grid and goes to the so-called "islanded mode", these inverters can provide a reference for the voltage frequency. However, an islanded microgrid consisting of only grid-forming inverters typically suffers from poor power sharing and desynchronized phases. In particular, one voltage source might bear a high percentage of the load, while the others provide a lower share of power. Furthermore, phase differences might lead to extreme voltage attenuation. "grid-feeding" inverters form the second category, working as current sources such that their voltage follow the frequency of the grid/microgrid. These sources are capable of injecting either a constant or a time-varying power regardless of the load, e.g. by the Maximum Power Point Tracking (MPPT) method. Inverters of this type are extensively used in the wide area electrical power networks. The third category contains sources that are designed to contribute to the regulation and stability of the microgrid, which are called "grid-supporting". This category includes both the Voltage Source Inverter (VSI) and the Current Source Inverter (CSI). Grid-supporting VSIs measure the active and reactive power they inject to the grid and determine the output voltage amplitude and frequency according to these measured values. Counter-wise, grid-supporting CSIs measure the voltage amplitude and frequency, in order to inject a desired amount of active and reactive power accordingly. Figure \ref{f:Itypes} depicts the inverter categories (for further discussion see \cite{Rocabert2012}).
\begin{figure}
  \includegraphics[width=8.5cm]{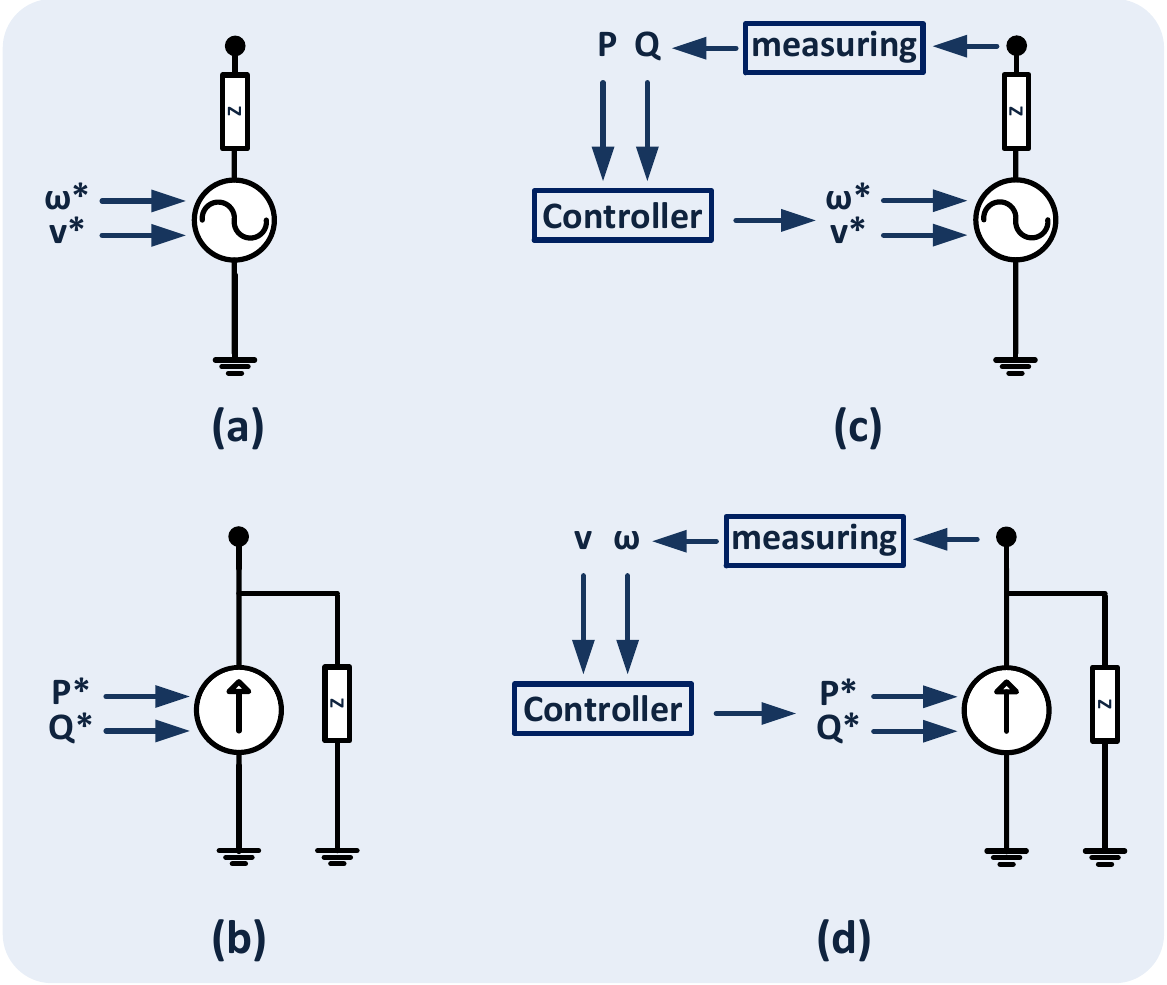}
  \caption{Inverter categories; (a) Grid-forming;  (b) Grid-feeding;  (c) Grid-supporting voltage source;  (d) Grid-supporting current source.}
  \label{f:Itypes}
\end{figure}

Controlling the grid-supporting inverters acting as VSI is mostly carried out by the design of a controller such that the sources artificially mimic the behavior and dynamics of a synchronous generator. This conveys that the frequency of the inverter deviates from the nominal value in case of a mismatch between the measured electrical output power and its nominal value. In particular, droop controllers are well known to generate such a behavior. Similar to synchronous generators, these inverters are capable of "self-regulation" and therefore the network consisting of these inverters synchronize in frequency \cite{SimpsonPorco2013,Schiffer2014,Nima2015}. Furthermore, the power injection of such sources are shared according to the droop coefficients, which are often proportional to the assigned nominal power values. 

While a large number of articles have focused on VSI, few have investigated the deployment of the  grid-supporting inverters acting as a CSI \cite{Rocabert2012}, \cite{Paquette2014}, \cite{Reigosa2009}, which are the key components in our proposed microgrid design. These devices are capable of injecting electrical power dynamically while following the frequency of the grid. Current source inverters inject active and reactive currents according to the rotating frame of the grid voltage at their point of connection. Therefore they need a voltage source online in the grid to follow its frequency. Controlling a grid with such a scheme has been first suggested by \cite{Chen1995} for operation of parallel uninterruptible power systems. However, the approach is not considered suitable for distributed generation (microgrid) mainly due to high communication requirements and the needs for supervisory control and extra cabling \cite{Hatziar2014}.

In this paper, we propose a microgrid structure, consisting of one synchronous generator (or a droop-controlled VSI) as the main power source, and a number of grid-supporting CSIs (see Figure \ref{fig:graph2}).
Inverters follow the frequency of the grid which is determined by the synchronous generator in the islanded mode. Such a system design is close to a master-slave architecture yet without communication. Note that no switching is needed to go through the transition from grid-connected to islanded and vice versa. Merely synchronizing the synchronous generator phase with the main grid suffices for going back to the grid-connected mode. In the case of the main generator failure, the so-called "tertiary control" can switch the master to another generator similar to the method of high-crest current \cite{Guerrero2008}. We show that our proposed system is stable and achieves the goal of power sharing for any topology of the network. Furthermore, we indicate that the system is capable of achieving cost optimization similar to the method used in \cite{Persis2014}. Additionally, as a step towards practical applicability, an implementation of the proposed inverter is also provided.

There are three main advantages of the proposed architecture over conventional droop-controlled VSIs. First, the voltage frequency of all sources are synchronized at all times, while droop-controlled VSIs located at different nodes, generate their own frequency. Second, communication-less secondary control can be reached with higher accuracy. Third and finally, there are physical advantages of using current source inverters over voltage source inverters such as smoother DC-side current, longer lifetime of energy storage, inherent voltage boosting capability, and lower costs \cite{Dash2011}. 

The following sections are organized as follows: In Section II, the structure and the control technique are elaborated. A model is suggested for the proposed microgrid and the system stability with the designed controller is investigated. Afterwards, power sharing of the sources and cost optimization are discussed. In the third section, an implementation of the proposed grid-supporting inverter is provided. Finally, Section 4 provides simulation results of the implemented model.
\begin{figure}
	\centering
	\includegraphics[width=6cm]{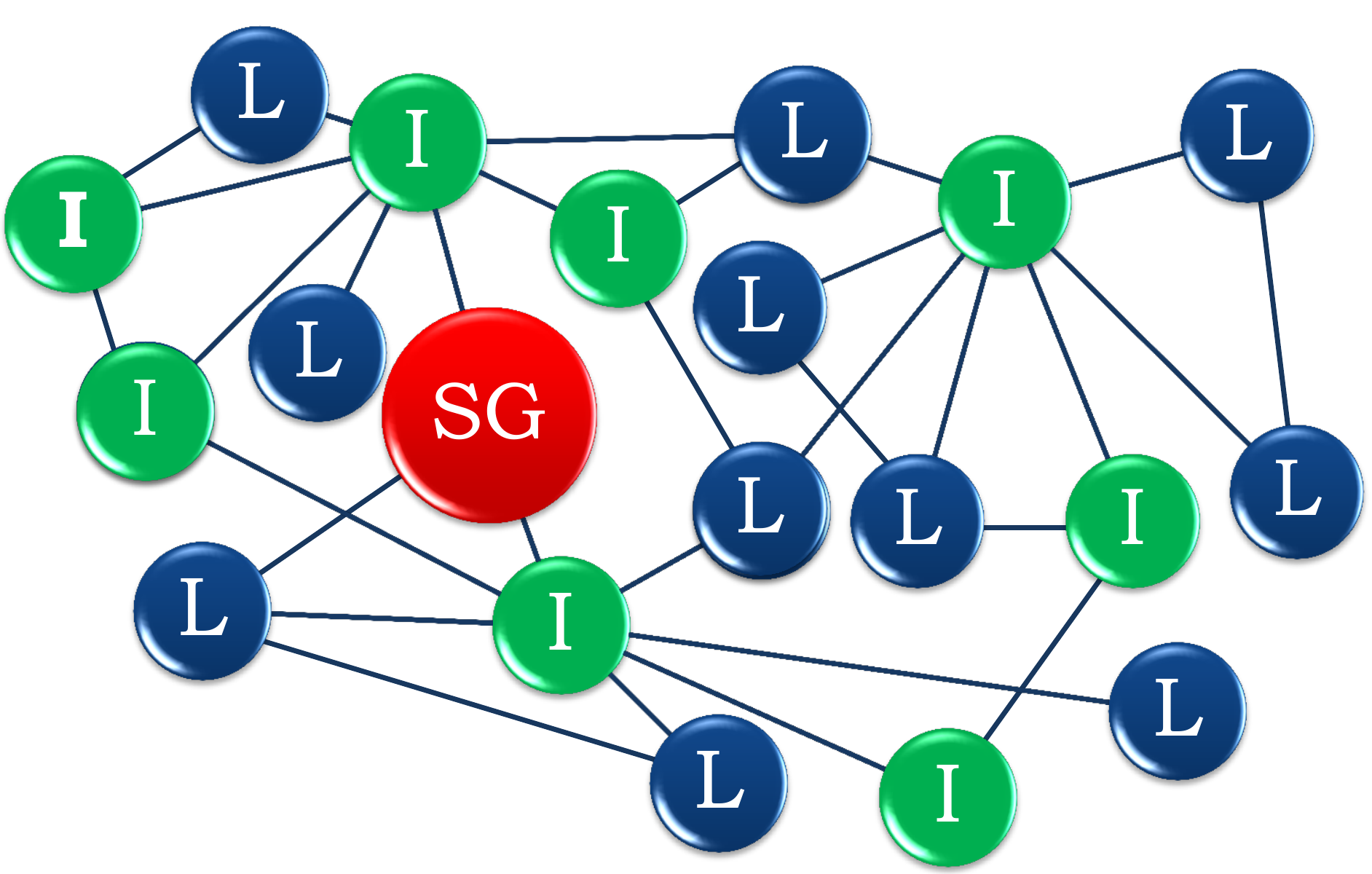}
	\caption{The proposed method applies to any topology of the microgrid network graph. The synchronous generator (SG) acts as the master, and inverters (I) act as slaves. The blue nodes represent power loads (L).}
	\label{fig:graph2}
\end{figure}
\section{Proposed Structure and Control Technique}
As mentioned in Section I, the microgrid in the proposed scheme contains a synchronous machine and multiple CSI inverters. The network of this grid is represented by a connected and undirected graph $\mathcal{G}(\mathcal{V},\mathcal{E} ) $. The nodes $\mathcal{V}= \{ 1,...,n \} $ represent a synchronous generator, grid supporting inverters and constant power loads. Edges $\mathcal{E} \subset \mathcal{V} \times \mathcal{V}$ account for the distribution lines. The nodes are partitioned as $\mathcal{V}=\mathcal{V}_G \cup \mathcal{V}_I \cup \mathcal{V}_L$, where $\mathcal{V}_G$, $\mathcal{V}_I$, and $\mathcal{V}_L$ correspond to the synchronous generator, inverters, and loads respectively. We also denote the cardinality of $\mathcal{V}_L$ and $\mathcal{V}_I$ with $n_I$ and $n_L$.
\subsection{Model}
We use the classical swing equation as the model of the SG \cite{Machowski2009}\footnote{D droop-controlled voltage source inverter can be modeled with similar dynamics, see \cite{Schiffer2014}.}
\begin{equation}\label{swing}
M\dot{\omega }+D\omega =P_m-P_e
  \; \text{,} 
\end{equation}
where $M \in \mathbb{R^+}$ is the angular momentum, $\omega \in \mathbb{R}$ is the frequency deviation from the nominal frequency, $D\in \mathbb{R^+}$ is the damping coefficient, and $P_m$ and $P_e$ are the mechanical (input) and electrical  (output) power respectively. The equality \eqref{swing} indicates that in case of a mismatch between mechanical power and electrical power, the synchronous generator rotor speed, as well as the voltage frequency, deviates from the nominal value.

Considering $P_{I,i}, P_{L,i} \in \mathbb{R}^+$ as the injected power of the inverter and the consumed power of the load at node $i$, we propose the following model for the aforementioned microgrid structure:
\begin{align*}
& M\dot{\omega }+D\omega =P_G^*+u-P_e 
\\
& P_{I,i}=P_{I,i}^*+\varUpsilon_i \numberthis \label{sys1}
\\
& P_{L,i}=P_{L,i}^*+\delta _{L,i}
\end{align*}
together with the coupling equality
\begin{align}
& P_e +\sum\limits_{i=1}^{n_I} P_{I,i}- \sum\limits_{i=1}^{n_L} P_{L,i}=0 \;\text{,} \label{powerbalance1} 
\end{align}
where $u \in \mathbb{R}$ is the control input of the synchronous generator, $\varUpsilon_i \in \mathbb{R}$ is the control input of the inverter at node $i$, and $\delta _{L,i} \in \mathbb{R}$, is the deviation from the nominal power consumption. We assume that the voltage magnitudes at each node in the network are constant and the reactive power injection at each node is identically zero. Note that the internal dynamics of $P_e$, $P_I$, and $P_L$ are not crucial for our analysis and their coupling equality \eqref{powerbalance1} represents the law of conservation of energy. Here, the distribution lines are assumed to be purely inductive implying that there is no active power consumption in the distribution lines. We also assume that nominal values $P_G^*$,$P_{I,i}^*$, and $P_{L,i}^*$ satisfy the power balance
\begin{equation}
P_G^*+\sum\limits_{i=1}^{n_I} P_{I,i}^*-\sum\limits_{i=1}^{n_L} P_{L,i}^*=0 \;\text{.} \label{powerbalance2}
\end{equation}
The system (\ref{sys1}) can be written in compact form as
\begin{align}
& M\dot{\omega }+D\omega =P_G^*+u-P_e \label{sys21}
\\
& P_I=P_I^*+\varUpsilon  \label{sys22}
\\
& P_L=P^*_L+\delta _L \; \text{,} \label{sys23}
\end{align}
where $\varUpsilon \in \mathbb{R}^{n_I}$ and $\delta_L \in \mathbb{R}^{n_L}$ are the inverter control input and load power vectors, which represent the deviations from the nominal power values. We consider the power sharing vector $\xi  \in \mathbb{R}^{n_I}$ to set the power injections of inverter controllers proportionally i.e. $\frac{\Upsilon_i}{\Upsilon_j}=\frac{\xi_i}{\xi_j}$. Without any loss of generality we assume that 
\begin{equation}
\1^T\xi=1\;\text{.} \label{xi}
\end{equation} 
Hence, we rewrite $\Upsilon$ as
\begin{align}
& \varUpsilon=v \xi \; \text{,} \label{Upsilon}
\end{align}
where $v \in \mathbb{R}$ is the common control signal for all inverters and is to be designed later. Bearing in mind \eqref{sys22}, \eqref{sys23}, and \eqref{Upsilon}, \eqref{powerbalance1} leads to
\begin{align*}
P_e&=\1^T(P_L^*+\delta _L)-\1^T(P_I^*+v\xi) \; \text{.}
\end{align*}
By substituting this into (\ref{sys21}), and using \eqref{xi} we have
\begin{equation*}
M\dot{\omega }+D\omega =P_G^*+u-\1^T(P_L^*+\delta _L)+\1^TP_I^*+v\ \; \text{.}
\end{equation*}
Finally, according to \eqref{powerbalance2}
\begin{equation}
M\dot{\omega }+D\omega =u+v-\1^T\delta _L \label{sys} \; \text{.}
\end{equation}
This indicates that the controllers at the inverter nodes are aligned with the action of synchronous generator controller.
\begin{remark}
Throughout the paper, the delays in the Phase Locked Loop (PLL) and Pulse Width Modulation (PWM) switchings are ignored and hence frequency following capability of inverters is assumed to be instantaneous. Also all line reactances are
assumed to be sufficiently low. Therefore, the network frequency $\omega$ has the same value and is available at all nodes. 
\end{remark}
\subsection{Controller Design}
\subsubsection{Controller for Inverters}
We consider the main generator  acting as the voltage reference with a fixed power control input ($\bar{u}$) and propose proportional integral controllers for the inverters\footnotemark:
\begin{align*} 
&  u=\bar{u}
\\
& \dot{\chi}=\omega  \numberthis \label{PI controller1} 
\\
& v=\underbrace{-\gamma \omega }_{v_1} \underbrace{-\beta \chi }_{v_2} \; \text{.}  
\end{align*}
\footnotetext{Although the integral of the voltage frequency can be interpreted as the voltage phase, we avoid denoting it by $\theta$, since here $\chi$ is a control variable and is not bounded to $[0,2\pi)$. }
Now, the system \eqref{sys} can be described as
\begin{align*} 
& M\dot{\omega }+(D+\gamma)\omega =v_2+\bar{u}-\1^T\delta _L
\\
&\dot{\chi }=\omega  \numberthis \label{PI system1}
\\
& v_2=-\beta \chi  \; \text{.}
\end{align*}
\begin{theorem}\label{theorem 1}
For any $\beta,\;\gamma>0$, the equilibrium $(\bar{\omega} =0,\bar{\chi }=\frac{\bar{u}-\1^T \delta _L}{\beta} )$ of system \eqref{PI system1} is globally asymptotically stable (GAS), and the inverter power injection $\Upsilon=\xi (-\gamma \omega-\beta \chi )$ converges asymptotically to $\bar{\Upsilon}=-\xi (\bar{u}-\1^T \delta _L )$.
\end{theorem} 
\begin{proof}
We consider the Lyapunov function
\begin{align*}
W=& \frac{1}{2}M\omega ^2 +\frac{1}{2}\beta( \chi -\bar{\chi })^2 \; \text{.}
\end{align*}
We have
\begin{align*}
\dot{W}=& -(D+\gamma) \omega ^2
\\& +\omega (v_2+\bar{u}-\1^T\delta _L)
\\& +(\beta \chi -(\bar{u}-\1^T\delta _L))\omega 
\\=& -(D+\gamma)\omega ^2 \; \text{.}
\end{align*}
 Observe that $W$ is radially unbounded and has a strict minimum at $\omega=0$ and $\chi =\frac{\bar{u}-\1^T \delta _L}{\beta}$. By invoking LaSalle invariance principle, the solutions converge to the largest invariant set for \eqref{PI system1} s.t. $\omega=0$. On this set, the solution to \eqref{PI system1} satisfy
\begin{align*}
&\dot{\chi }=0
\\
& \bar{v}_2=\1^T\delta _L \label{vbar}-\bar{u} \; \text{.} \numberthis
\end{align*}
The latter shows that all the solutions on the invariant set converge to the equilibrium ($\bar{\omega}=0$,$\bar{\chi} =\frac{\bar{u}-\1^T \delta _L}{\beta}$). Having shown before stability of this equilibrium, GAS of ($\bar{\omega}=0$,$\bar{\chi} =\frac{\bar{u}-\1^T \delta _L}{\beta}$) is guaranteed. Moreover, on the invariant set we have:
\begin{align*}
\bar{\Upsilon}=\xi (-\beta \bar{\chi})=-\xi (\bar{u}-\1^T \delta _L ) \; \text{.}
\end{align*}
This completes the proof.
\end{proof}
\begin{remark}
The interpretation of the proposed method is that by taking advantage of the information embedded in the frequency deviation as a result of the intrinsic droop characteristic of a synchronous generator, we avoid further communication. 
\end{remark}
\begin{remark}
In the sense of architectural control, the goals of primary and secondary controls are achieved. The advantage here, compared with the conventional droop controlled voltage source inverters, is that the integral controller is acting more precisely, since the integral action is taken over a common frequency at all nodes. Using integral controllers in the secondary control layer of a microgrid with conventional droop controllers, fails to maintain load sharing \cite{Dorfler2014} and is not considered a long time stable solution \cite{Guerrero2008}.
\end{remark}
Setting $\beta=0$, modifies the controller to a proportional one, i.e. $v=-\gamma \omega$. In this case, the system (\ref{PI system1}) can be described as
\begin{align*} \label{P system}
& M\dot{\omega }+(D+\gamma)\omega =\bar{u}-\1^T\delta _L
  \numberthis \; \text{.}
\end{align*}
\begin{corollary}
For any $\gamma>0$, the equilibrium $(\bar{\omega} =\frac{\bar{u}-\1^T\delta _L}{D+\gamma})$ of system \eqref{P system} is globally asymptotically stable.
\end{corollary}
\begin{proof}
We consider the quadratic Lyapunov function
\begin{align*}
W=& \frac{1}{2}M(\omega-\bar{\omega}) ^2 \; \text{.}
\end{align*}
We have
\begin{align*}
\dot{W}=& (\omega-\bar{\omega})(\bar{u}-\1^T\delta _L-(D+\gamma)\omega) 
\\=& -(D+\gamma)(\omega-\bar{\omega}) ^2 \; \text{,}
\end{align*}
which proves GAS of the equilibrium.
\end{proof}
It is observed that the inverters, here, contribute to the primary control level leading to a frequency deviation at steady state that is smaller than in the case of a sole generator.
\begin{remark}
In the case $\beta=0$, $v$ is equal to $-\gamma \omega$, and the inverter power injection at node $i$ is obtained as $P_{I,i}=P_{I,i}^*-\xi_i\gamma \omega$. The inverter dynamic model is hence similar to a conventional droop-controlled system. Here, the output power is regulated according to the frequency deviation of the grid, while in a droop-controlled voltage source inverter, the frequency is regulated according to the deviation from the nominal output power value. The advantage here is that the power network remains synchronized at all times. This leads to lower harmonics being received at the phased locked loops, by which voltage frequency (in our method) and output power (in conventional droop method) are measured.
\end{remark}
\subsubsection{Controllers for Inverters and the Main Generator}
Up to now constant power input of the controller ($\bar{u}$), is considered for the synchronous generator. Hence, no pre-determined sharing occurs between the main generator and the inverters. However, it is possible to make use of an integral controller for the main generator as well. This modifies the controller \eqref{PI controller1} to
\begin{align*} 
&\dot{\chi }=\omega 
\\
& u=-\alpha \chi  \numberthis \label{PI controller2}
\\
& v=\underbrace{-\gamma \omega }_{v_1} \underbrace{-\beta \chi }_{v_2} \; \text{.}
\end{align*}
Here, $\alpha\in \mathbb{R}^+$ and $\beta \in \mathbb{R}^+$ determine the power sharing at steady state between the synchronous generator and the rest of the sources (inverters). Let $\tilde{u}=u+v_2$, then the system can be described as
\begin{align*} 
& M\dot{\omega }+(D+\gamma)\omega =\tilde{u}-\1^T\delta _L
\\
&\dot{\chi }=\omega  \numberthis \label{PI system2}
\\
& \tilde{u}=-(\alpha+\beta) \chi  \; \text{.}
\end{align*}
\begin{corollary} \label{SGI}
The equilibrium $(\bar{\omega} =0,\bar{\chi }=-\frac{1}{\alpha+\beta} \1^T \delta _L)$ of system \eqref{PI system2} is globally asymptotically stable.
\end{corollary}
\begin{proof}
We consider the Lyapunov  function
\begin{align*}
W=& \frac{1}{2}M\omega ^2 +\frac{1}{2}(\alpha+\beta)( \chi -\bar{\chi })^2 \; \text{.}
\end{align*}
We have
\begin{align*}
\dot{W}=& -(D+\gamma) \omega ^2
\\& +\omega (\tilde{u}-\1^T\delta _L)
\\& +((\alpha+\beta) \chi +\1^T\delta _L)\omega 
\\=& -(D+\gamma)\omega ^2 \; \text{.}
\end{align*}
The GAS of the equilibrium follows analogous to the proof of Theorem \ref{theorem 1}. 
\end{proof}
\begin{remark}
System \eqref{PI system2} provides stronger and faster controller overcoming limitations in controlling the power fluctuations with a sole synchronous generator. 
\end{remark}
\subsection{Power Sharing}
\subsubsection{Inverters}
The inverters described by model \eqref{P system} contribute to decrease the frequency deviation.  However, the inverters provided in model \eqref{PI system1} modify the power injection till the frequency of the microgrid is regulated and hence reaches the steady state $\omega =0$. In both cases, the increase in power injections are proportional according to the sharing vector $\xi$. 

Provided that some costs are assigned to the power injection of the inverters, optimization methods can be adapted to minimize the energy cost. Here, we show that by a proper choice of the sharing vector $\xi $ cost optimization can be achieved.
We consider the quadratic cost function
\begin{align*}
C(\bar{\Upsilon})=\frac{1}{2} \bar{\Upsilon}^T\Lambda\bar{\Upsilon} \; \text{,}
\end{align*}
where $\Lambda \in \mathbb{R}^{n_I\times n_I}$ is a diagonal matrix of cost coefficients assigned to each node (inverter). Recall that, by Theorem \ref{theorem 1}, at steady-state we have
\[
\bar{\Upsilon}=-\xi (\bar{u}-\1^T \delta _L ).
\]
Multiplying both sides of the above equality from the left by $\1^T$ yields the following optimization constraint
\begin{align*}
0=\bar{u}+\1^T \bar{\Upsilon} - \1^T\delta_L \;\text{.}
\end{align*}
To minimize the Lagrangian function 
\begin{align*}
\mathcal{L}=\frac{1}{2} \bar{\Upsilon}^T\Lambda\bar{\Upsilon}+\mu(\bar{u}+\1^T \bar{\Upsilon} - \1^T\delta_L) \; \text{,}
\end{align*} 
we have $$\mu=\frac{\bar{u}-\1^T\delta_L}{\1^T\Lambda^{-1}\1}.$$ Hence,
\begin{align*}
\bar{\Upsilon}&=\frac{\1^T\delta_L-\bar{u}}{\1^T\Lambda^{-1}\1} \Lambda^{-1}\1.
\end{align*}
In accordance with \eqref{Upsilon}, the vector $\bar\Upsilon$ can be rewritten as 
\begin{align}\label{Nima}
\bar{\Upsilon}&=(\1^T\delta_L-\bar{u})\xi_{opt}
\end{align}
where
\begin{align} \label{opt}
	\xi_{opt}=\frac{\Lambda^{-1}\1}{\1^T\Lambda^{-1}\1}.
\end{align}  
Noting that \eqref{opt} satisfies \eqref{xi}, we have the following corollary. 
\begin{corollary}
Let the vector $\xi_{opt}$ be given by \eqref{opt}. Then any solution to system \eqref{PI system1} is such that the corresponding inverter power injection $\Upsilon=v\xi_{opt}$ asymptotically converges to the optimal power injection \eqref{Nima}.
\end{corollary}
\begin{proof}
Since $\Upsilon=v\xi_{opt}$, and $v=-\gamma \omega - \beta \chi$, we know from Theorem \ref{theorem 1} that $\Upsilon$ converges to $-\beta \bar{\chi} \xi_{opt}$, which is equal to $ -(\bar{u}-\1^T \delta_L)\xi_{opt}$. This ends the proof.
\end{proof}
\begin{remark}
Similar to adjusting droop coefficients in conventional droop methods, various parameters can as well be interpreted (directly or inversely) as cost to form $\Lambda$ in \eqref{opt}: nominal power values (power ratings) \cite{Chandorkar1993}, State of Charge (Soc) in storage systems \cite{Guerrero2009}, insolation level in photovoltaic (PV) systems, or wind power in wind turbines. Alike the established droop methods, these values can be transmitted to the nodes via a low-bandwidth communication (tertiary control).
\end{remark}
\begin{remark}
Provided that the inverse of nominal power values are considered as cost coefficients, i.e. $P^*_I=\Lambda^{-1} \1$, we have
\begin{align*}
 \xi^*_{opt}=\frac{P^*_I}{\sum\limits_{i=1}^{n_I}P^*_I} \; \text{.}
\end{align*}
Therefore
\begin{align*}
\frac{P_{I,i}}{P_{I,j}}=\frac{P^*_{I,i}+\Upsilon_i}{P^*_{I,j}+\Upsilon_j}=\frac{\xi^*_{opt,i}}{\xi^*_{opt,j}} \; \text{.}
\end{align*}
This implies that in this case the total power injection (nominal+deviation) of each inverter remains proportional to the others.
\end{remark}
\subsubsection{Synchronous Generator and Inverters}
The main generator and the inverters share power in the model \eqref{PI system2}. All changes in the load are shared between the generator and the inverters according to the $\alpha$ and $\beta$ in the steady state. Using  \eqref{PI system2} from Corollary \ref{SGI} we have
\begin{align*}
& \bar{u}=(\frac{\alpha}{\alpha+\beta})\1^T\delta _L
\\
& \bar{v}=(\frac{\beta}{\alpha+\beta})\1^T\delta _L
\\
& \frac{\bar{u}}{\bar{v}}=\frac{\alpha}{\beta} \; \text{,}
\end{align*}
which implies that the deviations from the nominal power consumption will be shared according to $\alpha/\beta$ ratio between the synchronous generator and the inverters. 
\begin{remark}
If the ratio $\alpha/\beta$ is set equal to the ratio of the nominal power of the synchronous generator to the sum of the inverter nominal power values, i.e.,  $\frac{\alpha}{\beta}=\frac{P_G^*}{\sum\limits_{i=1}^{n_I} P^*_{I,i}}$, we have
\begin{equation*}
\frac{P_G^*+\bar{u}}{\sum\limits_{i=1}^{n_I} P^*_{I,i}+\bar{v}}=\frac{\alpha}{\beta} \; \text{,}
\end{equation*} 
which implies that the total power injections remain proportional.
\end{remark}
\begin{remark}
The ratio $P_G^*/\sum\limits_{i=1}^{n_I} P^*_{I,i}$ can be modified without the need of any communication. In fact, power balance at any steady state $\omega =0$, can define new nominal values. Hence the input power of the synchronous generator ($u$) can be independently increased/decreased to a value which the network would achieve if the load has decreased/increased. In this case, the rate of change in power dispatch (control input $u$) for the synchronous machine can be as slow as the machine can safely follow without rapid fluctuations in frequency. This change in  (master) power dispatch will be automatically followed by changes in the injected power of the inverters (slaves). 
\end{remark}
\begin{figure}
	\centering
	\includegraphics[width=8.5cm]{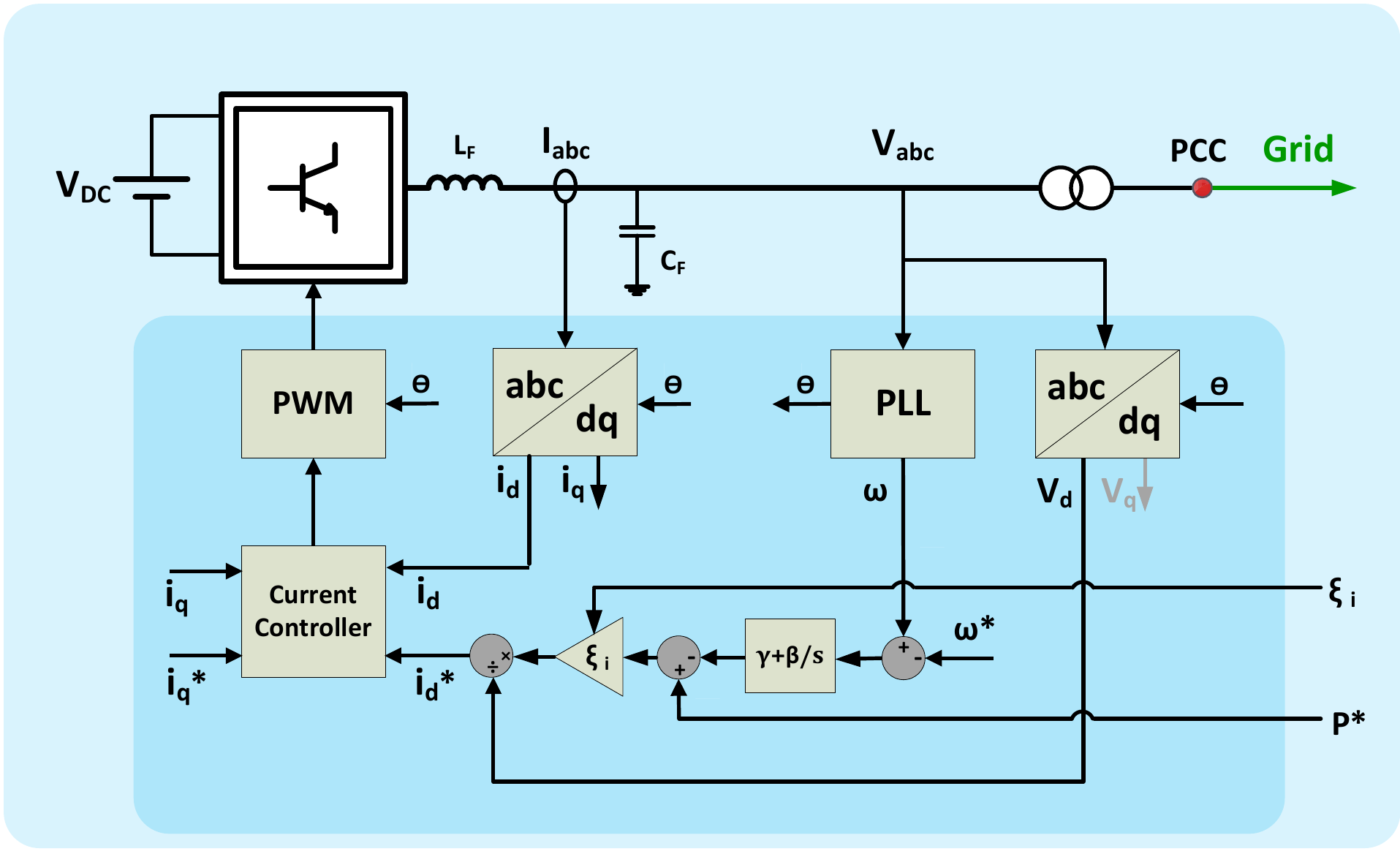}
	\caption{Implementation of the proposed grid-supporting inverter.}
	\label{imp}
\end{figure}
\section{Implementation}
Figure \ref{imp} depicts an implementation of the proposed grid-supporting current source inverter. The structure is similar to the CSI schematic provided in \cite{Rocabert2012} with modifications in order to generate desired current value $i^*$ according to our controller design. The instantaneous power delivered by the inverter is
\begin{equation}\label{instantp}
P=v_di_d+v_qi_q\;,
\end{equation}
where $v_d,v_q$ and $i_d,i_q$ are the direct and quadrature components of output current and voltage, resulting from the $dq0$ transformation. Since, the rotating frame is synchronized with the voltage, the quadrature component $v_q$ is kept at zero. Hence, \eqref{instantp} reads as
\begin{equation*}
P=v_di_d\;.
\end{equation*}
Therefore by setting $i_d$ to a desired $i^*_d$, it is possible to control the injected power. The PLL measures the frequency and the deviation from the nominal frequency is given as input to the controller. This input is added to the nominal power value $P^*$ (refer to the second equation in ), assigned a gain proportional to the power sharing, and divided by the direct axis voltage to generate the desired current $i^*_d$. PI controllers are used in the current controller block to set the output current $i_d$ to $i^*_d$.
\begin{figure}[t]
\centering
  \includegraphics[width=8.5cm]{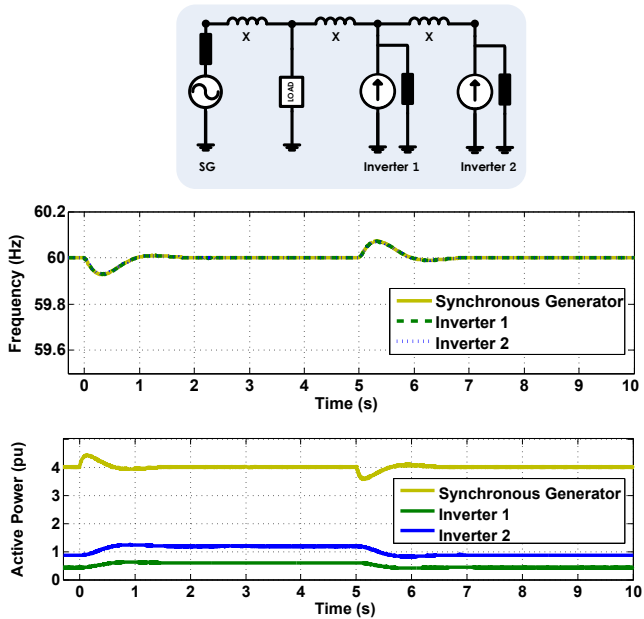}
  \caption{Simulation schematic and results.}
  \label{f:results}
\end{figure}
\section{Simulation}
The performance of the implemented model is simulated in SimPowerSystems, Matlab Simulink. A synchronous generator is connected to two current source inverters in series via inductive lines. The input power of the main generator is kept constant while the inverters are provided with integral controllers (see model \eqref{PI controller1}).
Figure \ref{f:results} depicts the schematic of the simulated microgrid, and frequency regulation and power sharing of the inverters. There is an increase in the load power ($0.5$ pu) at $t=0$ and is set back to the previous value at $t=5$. Results depict that the system is precisely synchronized in frequency at all times. Furthermore, inverters keep their ratio in sharing the active power and return precisely to their original power injection after the load is relieved.
The parameters (per unit) used in this simulation are as follows:
Synchronous Generator: $M=0.1$, $D=0.05$, $P^*_G=1$, $u=3$
; 
Inverters: $\gamma=0.15$, $\beta=1.5$, $P^*_I=1.5$, $\xi_1=\frac{1}{3}$, $\xi_2=\frac{2}{3}$
; 
Load: $P^*_L=5.5$, $\delta_L=0.5$
;
Line Inductance: $X=0.12$.
\section{Conclusion}
A microgrid with a master-slave architecture is proposed in which a synchronous generator acts as the master and current source grid-supporting inverters act as the slaves. A model is provided for such a system and its stability is analyzed. Furthermore, an implementation for the inverter with the proposed controller is provided. Results confirm the validity of the modeling and the capability of the proposed microgrid to achieve power sharing and frequency regulation without any communication among sources. Future works include considering the dynamics of the phase locked loop, reactive power sharing, and complex distribution lines.
\section{Acknowledgment}
The authors would like to thank Mathijs Lomme for providing Matlab Simulink model of current source inverters (CSI).
\bibliographystyle{IEEETran}
\bibliography{ref}
\IEEEpeerreviewmaketitle
\end{document}